\documentclass[copyright]{eptcs}
\providecommand{\event}{DCFS 2010} 
\usepackage{breakurl}             

\usepackage{amsmath}
\usepackage{amssymb}
\usepackage{graphics}
\usepackage{amsthm}

\newtheorem{theorem}{Theorem}

\title{Graph-Controlled Insertion-Deletion Systems}

\author{Rudolf Freund
\institute{Faculty of Informatics, Vienna University of Technology,
Favoritenstr. 9, 1040 Vienna, Austria
}
\email{rudi@emcc.at}
\and
Marian Kogler
\institute{Faculty of Informatics, Vienna University of Technology,
Favoritenstr. 9, 1040 Vienna, Austria\\
Institute of Computer Science, Martin Luther University Halle-Wittenberg\\
Von-Seckendorff-Platz 1, 06120 Halle (Saale), Germany
\email{marian@emcc.at, kogler@informatik.uni-halle.de}
}
\and
Yurii Rogozhin
\institute{
Institute of Mathematics and Computer Science,
Academy of Sciences of Moldova\\
Str. Academiei 5, Chi\c sin\u au, MD-2028, Moldova}
\email{rogozhin@math.md}
\and
Sergey Verlan
\institute{
LACL, D\'{e}partement Informatique,
UFR Sciences et Technologie\\
Universit\'{e} Paris Est,
61, av. G\'{e}n\'{e}ral de Gaulle, 94010 Cr\'{e}teil, France}
\email{verlan@univ-paris12.fr}
}

\def\titlerunning{Graph-Controlled Insertion-Deletion Systems}
\def\authorrunning{R.~Freund, M.~Kogler, Y.~Rogozhin, S.~Verlan}

\begin{document}

\maketitle

\begin{abstract}
In this article, we consider the operations of insertion and deletion
working in a graph-controlled manner. We show that like in the case of
context-free productions, the computational power is strictly increased when
using a control graph: computational completeness can be obtained by systems
with insertion or deletion rules involving at most two symbols in a
contextual or in a context-free manner and with the control graph having
only~four nodes.
\end{abstract}

\section{Introduction}

The operations of insertion and deletion were first considered with a
linguistic motivation~\cite{Marcus,Galiuk,Kluwer}. Another inspiration for
these operations comes from the fact that the insertion operation and its
iterated variants are generalized versions of Kleene's operations of
concatenation and closure~\cite{Kleene56}, while the deletion operation
generalizes the quotient operation. A study of properties of the
corresponding operations may be found in~\cite{Haussler82,Haussler83,Kari}.
Insertion and deletion also have interesting biological
motivations, e.g., they correspond to a mismatched annealing of DNA
sequences; these operations are also present in the evolution processes in
the form of point mutations as well as in RNA editing, see the discussions
in~\cite{Beene,BBD07,Smith} and~\cite{dna}. These biological motivations of
insertion-deletion operations led to their study in the framework of
molecular computing, see, for example, \cite{Daley,cross,dna,TY}.

In general, an insertion operation means adding a substring to a given
string in a specified (left and right) context, while a deletion operation
means removing a substring of a given string from a specified (left and
right) context. A finite set of insertion-deletion rules, together with a
set of axioms provide a language generating device: starting from the set of
initial strings and iterating insertion-deletion operations as defined by
the given rules, one obtains a language.

Even in their basic variants, insertion-deletion systems are able to
characterize the recursively enumerable languages. Moreover, as it was shown
in \cite{cfinsdel}, the context dependency may be replaced by insertion and
deletion of strings of sufficient length, in a context-free manner. If the
length is not sufficient (less or equal to two) then such systems are
not able to generate more than the recursive languages and a characterization of them was shown in~\cite{SV2-2}.

Similar investigations were continued in \cite{MRV07,KRV08,KRV08c} on
insertion-deletion systems with one-sided contexts, i.e., where the context
dependency is present only from the left or only from the right side of all
insertion and deletion rules. The papers cited above give several
computational completeness results depending on the size of 
insertion and deletion rules. We recall the interesting fact that some
combinations are not leading to computational completeness, i.e., there are
recursively enumerable languages that cannot be generated by such devices.

Like in the case of context-free rewriting, it is possible to consider a
graph-controlled variant of insertion-deletion systems. Thus the rules
cannot be applied at any time, as their applicability depends on the current
\textquotedblleft state\textquotedblright , changed by a rule application.
There are several equivalent definitions of a graph-controlled application,
we consider one of them where rules are grouped in \emph{components}
(communication graph nodes) and at each step one of the rules from the
current component is chosen non-deterministically and applied, at the same
time changing the current component. Such a formalization is rather similar
to the definition of insertion-deletion P systems~\cite{membr}, however it
is even simpler and more natural.

This article focuses on one-sided graph-controlled insertion-deletion
systems where at most two symbols may be present in the description of
insertion and deletion rules. This correspond to systems of size $%
(1,1,0;1,1,0)$, $(1,1,0;1,0,1)$, $(1,1,0;2,0,0)$, and $(2,0,0;1,1,0)$, where
the first three numbers represent the maximal size of the inserted string
and the maximal size of the left and right contexts, while the last three
numbers represent the same information, but for deletion rules. It is known
that such systems are not computationally complete~\cite{KRV09}, while the
corresponding P systems variants are computationally complete, which results
are achieved with five components. In this article we give a simpler
definition of the concept of graph-controlled insertion-deletion systems and
we show that computational completeness can already be achieved by using a
control graph with only four nodes (components).

\section{Definitions}

We do not present the usual definitions concerning standard concepts of the
theory of formal languages and we only refer to textbooks such as \cite%
{handbook} for more details.

The empty string is denoted by $\lambda $. For the interval of natural
numbers from $k$ to $m$ we write $\left[ k..m\right] $.

In the following, we will use special variants of the \emph{Geffert} normal
form for type-0 grammars (see~\cite{Geffert91} for more details).

A grammar $G=\left( N,T,P,S\right) $ is said to be in \emph{Geffert normal
form}~\cite{Geffert91} if $N=\{S,A,B,C,D\}$ and $P$ only contains
context-free rules of the forms $S\rightarrow uSv$ with 
$u\in \{A,C\}^{\ast }$ and $v\in \{B,D\}^{\ast }$ as well as 
$S\rightarrow x$ with $x\in (T\cup \{A,B,C,D\})^{\ast }$ and two 
(non-context-free) erasing rules $AB\rightarrow \lambda $ and 
$CD\rightarrow \lambda $.

We remark that we can easily transform the linear rules from the Geffert
normal form into a set of left-linear and right-linear rules (by increasing
the number of non-terminal symbols, e.g., see \cite{membr}). More
precisely, we say that a grammar $G=\left( N,T,P,S\right) $ with 
$N=N^{\prime }\cup N^{\prime \prime }$, $S,S^{\prime }\in N^{\prime }$, and $N^{\prime \prime }=\{A,B,C,D\}$, is in the \emph{special Geffert normal 
form} if, besides the two erasing rules $AB\rightarrow \lambda $ and 
$CD\rightarrow \lambda $, it only has context-free rules of the following 
forms:

\begin{align*}
& X\rightarrow bY,\quad X,Y\in N^{\prime },b\in T\cup N^{\prime \prime }, \\
& X\rightarrow Yb,\quad X,Y\in N^{\prime },b\in T\cup N^{\prime \prime }, \\
& S^{\prime }\rightarrow \lambda .
\end{align*}

Moreover, we may even assume that, except for the rules of the forms $%
X\rightarrow Sb$ and $X\rightarrow S^{\prime }b$, for the first two types of
rules it holds that the right-hand side is unique, i.e., for any two rules $%
X\rightarrow w$ and $U\rightarrow w$ in $P$ we have $U=X$.

The computation in a grammar in the special Geffert normal form is done in
two stages. During the first stage, only context-free rules are applied.
During the second stage, only the erasing rules $AB\rightarrow \lambda $ and 
$CD\rightarrow \lambda $ are applied. These two erasing rules are not
applicable during the first stage as long as the left and the right part of
the current string are still separated by $S$ (or $S^{\prime }$) as all the 
symbols $A$ and $C$ are generated on the left side of these middle symbols 
and the corresponding symbols $B$ and $D$ are generated on the right side. 
The transition between stages is done by the rule $S^{\prime }\rightarrow
\lambda $. We remark that all these features of a grammar in the special
Geffert normal form are immediate consequences of the proofs given in~\cite%
{Geffert91}.

\subsection{Insertion-deletion systems}

An \textit{insertion-deletion system} is a construct $ID=(V,T,A,I,D)$, where 
$V$ is an alphabet; $T\subseteq V$ is the set of \textit{terminal} symbols
(in contrast, those of $V-T$ are called \textit{non-terminal} symbols); $A$
is a finite language over $V$, the strings in $A$ are the \textit{axioms}; $%
I,D$ are finite sets of triples of the form $(u,\alpha ,v)$, where $u$, $%
\alpha $ ($\alpha \neq \lambda $), and $v$ are strings over $V$. The triples
in $I$ are \textit{insertion rules}, and those in $D$ are \textit{deletion
rules}. An insertion rule $(u,\alpha ,v)\in I$ indicates that the string $%
\alpha $ can be inserted between $u$ and $v$, while a deletion rule $%
(u,\alpha ,v)\in D$ indicates that $\alpha $ can be removed from between the
context $u$ and $v$. Stated in another way, $(u,\alpha ,v)\in I$ corresponds
to the rewriting rule $uv\rightarrow u\alpha v$, and $(u,\alpha ,v)\in D$
corresponds to the rewriting rule $u\alpha v\rightarrow uv$. By $%
\Longrightarrow _{ins}$ we denote the relation defined by the insertion
rules (formally, $x\Longrightarrow _{ins}y$ if and only if $%
x=x_{1}uvx_{2},y=x_{1}u\alpha vx_{2}$, for some $(u,\alpha ,v)\in I$ and $%
x_{1},x_{2}\in V^{\ast }$), and by $\Longrightarrow _{del}$ the relation
defined by the deletion rules (formally, $x\Longrightarrow _{del}y$ if and
only if $x=x_{1}u\alpha vx_{2},y=x_{1}uvx_{2}$, for some $(u,\alpha ,v)\in D$
and $x_{1},x_{2}\in V^{\ast }$). By $\Longrightarrow $ we refer to any of
the relations $\Longrightarrow _{ins},\Longrightarrow _{del}$, and by $%
\Longrightarrow ^{\ast }$ we denote the reflexive and transitive closure of $%
\Longrightarrow $.

The language generated by $ID$ is defined by 
\begin{equation*}
L(ID)=\{w\in T^{\ast }\mid x\Longrightarrow ^{\ast }w\mathrm{\ for\ some\ }%
x\in A\}.
\end{equation*}

The complexity of an insertion-deletion system $ID=(V,T,A,I,D)$ is described
by the vector $$(n,m,m^{\prime };p,q,q^{\prime })$$ called \emph{size}, where 
\vspace{-2mm} 
\begin{eqnarray*}
n=\max\{|\alpha|\mid (u,\alpha,v)\in I\}, & & p=\max\{|\alpha|\mid
(u,\alpha,v)\in D\}, \\
m=\max\{|u|\mid (u,\alpha,v)\in I\}, & & q=\max\{|u|\mid (u,\alpha,v)\in D\},
\\
m^{\prime }=\max\{|v|\mid (u,\alpha,v)\in I\}, & & q^{\prime }=\max\{|v|\mid
(u,\alpha,v)\in D\}.
\end{eqnarray*}

By $INS_{n}^{m,m^{\prime }}DEL_{p}^{q,q^{\prime }}$ we denote the families
of insertion-deletion systems having size $(n,m,m^{\prime
};p,q,q^{\prime })$.

If one of the parameters $n,m,m^{\prime },p,q,q^{\prime }$ is not specified,
then instead we write the symbol~$\ast $. In particular, $INS_{\ast
}^{0,0}DEL_{\ast }^{0,0}$ denotes the family of languages generated by \emph{%
context-free insertion-deletion systems}. If one of numbers from the pairs $%
m $, $m^{\prime }$ and/or $q$, $q^{\prime }$ is equal to zero (while the
other one is not), then we say that the corresponding families have a
one-sided context. Finally we remark that the rules from $I$ and $D$ can be
put together into one set of rules $R$ by writing $\left( u,\alpha ,v\right)
_{ins}$ for $\left( u,\alpha ,v\right) \in I$ and $\left( u,\alpha ,v\right)
_{del}$ for $\left( u,\alpha ,v\right) \in D$.

\subsection{Graph-controlled insertion-deletion systems}

Like context-free grammars, insertion-deletion systems may be extended by
adding some additional controls. We discuss here the adaptation of the idea
of programmed and graph-controlled grammars for insertion-deletion systems.

A \emph{graph-controlled insertion-deletion system} is a construct%
\begin{equation*}
\Pi =(V,T,A,H,I_{0},I_{f},R)\mathrm{\ where}
\end{equation*}

\begin{itemize}
\item $V$ is a finite alphabet,

\item $T\subseteq V$ is the \emph{terminal alphabet},

\item $A\subseteq V^{\ast }$ is a finite set of \emph{axioms},

\item $H$ is a set of labels associated (in a one-to-one manner) to the
rules in $R$,

\item $I_{0}\subseteq H$ is the set of \emph{initial labels},

\item $I_{f}\subseteq H$ is the set of \emph{final labels}, and

\item $R\ $is a finite set of rules of the form $l:\left( r,E\right) $ where 
$r$ is an insertion or deletion rule over $V$ and $E\subseteq H$.
\end{itemize}

As is common for graph controlled systems, a configuration of $\Pi $ is
represented by a pair $(i,w)$, where $i$ is the label of the rule to be
applied and $w$ is the current string. A transition $(i,w)\Rrightarrow
(j,w^{\prime })$ is performed if there is a rule $l:\left( \left( u,\alpha
,v\right) _{t},E\right) $ in $R$ such that $w\Longrightarrow _{t}w^{\prime }$
by the insertion/deletion rule $(u,\alpha ,v)_{t}$, $t\in \left\{
ins,del\right\} $, and $j\in E$. The result of the computation consists of
all terminal strings reaching a final label from an axiom and the initial
label, i.e., 
\begin{equation*}
L(\Pi )=\{w\in T^{\ast }\mid (i_{0},w^{\prime })\Rrightarrow ^{\ast
}(i_{f},w)\mathrm{\ for\ some\ }w^{\prime }\in A,\mathrm{\ }i_{0}\in I_{0},%
\mathrm{\ }i_{f}\in I_{f}\}.
\end{equation*}

We may use another rather similar definiton for a graph-controlled
insertion-deletion system, thereby assigning groups of rules to \textit{%
components} of the system:

A \emph{graph-controlled insertion-deletion system with }$\emph{k}$\emph{\
components} is a construct%
\begin{equation*}
\Pi =(k,V,T,A,H,i_{0},i_{f},R)\mathrm{\ where}
\end{equation*}

\begin{itemize}
\item $k$ is the number of components,

\item $V,T,A,H$ are defined as for graph-controlled insertion-deletion
systems,

\item $i_{0}\in \left[ 1..k\right] $ is the initial component,

\item $i_{f}\in \left[ 1..k\right] $ is the final component, and

\item $R\ $is a finite set of rules of the form $l:\left( i,r,j\right) $
where $r$ is an insertion or deletion rule over $V$ and $i,j\in \left[ 1..k%
\right] $.
\end{itemize}

The set of rules $R$ may be divided into sets $R_{i}$ assigned to the \emph{%
components }$i\in \left[ 1..k\right] $, i.e., $R_{i}=\left\{ l:\left(
r,j\right) \mid l:\left( i,r,j\right) \in R\right\} $; in a rule $l:\left(
i,r,j\right) $, the number $j$ specifies the \emph{target component} where
the string is sent from component $i$ after the application of the insertion
or deletion rule $r$. A configuration of $\Pi $ is represented by a pair $%
(i,w)$, where $i$ is the number of the \emph{current} component (initially $%
i_{0}$) and $w$ is the current string. We also say that $w$ is \emph{situated%
} in component $i$. A transition $(i,w)\Rrightarrow (j,w^{\prime })$ is
performed as follows: first, a rule $l:\left( r,j\right) $ from component $i$
(from the set $R_{i}$) is chosen in a non-deterministic way, the rule $r$ is
applied, and the string is moved to component $j$; hence, the new set from
which the next rule to be applied will be chosen is $R_{j}$. More formally, $%
(i,w)\Rrightarrow (j,w^{\prime })$ if there is $l:\left( \left( u,\alpha
,v\right) _{t},j\right) \in R_{i}$ such that $w\Longrightarrow _{t}w^{\prime
}$ by the rule $\left( u,\alpha ,v\right) _{t}$; we also write $%
(i,w)\Rrightarrow _{l}(j,w^{\prime })$ in this case. The result of the
computation consists of all terminal strings situated in component $i_{f}$
reachable from the axiom and the initial component, i.e., 
\begin{equation*}
L(\Pi )=\{w\in T^{\ast }\mid (i_{0},w^{\prime })\Rrightarrow ^{\ast
}(i_{f},w)\mathrm{\ for\ some\ }w^{\prime }\in A\}.
\end{equation*}

Is is not difficult to see that graph-controlled insertion-deletion systems
with $k$ components are a special case of graph-controlled
insertion-deletion systems. Without going into technical details, we just
give the main ideas how to obtain a graph-controlled insertion-deletion
system from a graph-controlled insertion-deletion system with $k$
components: for every $l:\left( \left( u,\alpha ,v\right) _{t},j\right) \in
R_{i}$ we take a rule $l:\left( i,\left( u,\alpha ,v\right) _{t},Lab\left(
R_{j}\right) \right) $ into $R$ where $Lab\left( R_{j}\right) $ denotes the
set of labels for the rules in $R_{j}$; moreover, we take $I_{0}=Lab\left(
R_{i_{0}}\right) $ and $I_{f}=Lab\left( R_{i_{f}}\right) $. Finally, we
remark that the labels in a graph-controlled insertion-deletion system with $%
k$ components may even be omitted, but they are useful for specific proof
constructions. On the other hand, by a standard powerset construction for
the labels (as used for the determinization of non-deterministic finite
automata) we can easily prove the converse inclusion, i.e., that for any
graph-controlled insertion-deletion system we can construct an equivalent
graph-controlled insertion-deletion system with $k$ components.

We define the \emph{communication graph} of a graph-controlled
insertion-deletion system with $k$ components to be the graph with nodes $%
1,\dots {},k$ having an edge from node $i$ to node $j$ if and only if there
exists a rule $l:\left( \left( u,\alpha ,v\right) _{t},j\right) \in R_{i}$.
In~\cite{membr}, 5.5, special emphasis is laid on graph-controlled
insertion-deletion systems with $k$ components whose communication graph has
a tree structure, as we observe that the presentation of graph-controlled
insertion-deletion systems with $k$ components given above in the case of a
tree structure is rather similar to the definition of insertion-deletion P
systems as given in \cite{membr}; the main differences are that in P systems
the final component $i_{f}$ contains no rules and corresponds with the root
of the communication tree; on the other hand, in graph-controlled
insertion-deletion system with $k$\ components, each of the axioms can only
be situated in the initial component $i_{0}$, whereas in P systems we may
situate each axiom in various different components.

Throughout the rest of this paper we shall only use the notion of
graph-controlled insertion-deletion systems with $k$ components, as they are
easier to handle and sufficient to establish computational completeness in
the proofs of our main results presented in the succeeding section. By $%
GCL_{k}(ins_{n}^{m,m^{\prime }},del_{p}^{q,q^{\prime }})$ we denote the
family of languages $L(\Pi )$ generated by graph-controlled
insertion-deletion systems with at most $k$ components and insertion and
deletion rules of size at most $(n,m,m^{\prime };p,q,q^{\prime })$. We
replace $k$ by $\ast $ if $k$ is not fixed. The letter \textquotedblleft
G\textquotedblright\ is replaced by the letter \textquotedblleft
T\textquotedblright\ to denote classes whose communication graph has a \emph{%
tree structure}. Some results for the families $TCL_{k}(ins_{n}^{m,m^{\prime
}},del_{p}^{q,q^{\prime }})$ can directly be derived from the results
presented in~\cite{KRV09,membr} for the corresponding families of
insertion-deletion P systems $ELSP_{k}(ins_{n}^{m,m^{\prime
}},del_{p}^{q,q^{\prime }})$, yet the results we present in the succeeding
section either reduce the number of components for systems with an
underlying tree structure or else take advantage of the arbitrary structure
of the underlying communication graph thus obtaining computational
completeness for new restricted variants of insertion and deletion rules.

\section{Main results}

For all the variants of insertion and deletion rules considered in this
section, we know that the basic variants without using control graphs cannot achieve computational completeness (see \cite{KRV09}, \cite{MRV07}). The
computational completeness results from this section are based on
simulations of derivations of a grammar in the special Geffert normal form.
These simulations associate a group of insertion and deletion rules to each
of the right- or left-linear rules $X\rightarrow bY$ and $X\rightarrow Yb$.
The same holds for (non-context-free) erasing rules $AB\rightarrow \lambda $
and $CD\rightarrow \lambda $. We remark that during the derivation of a
grammar in the special Geffert normal form, any sentential form contains at
most one non-terminal symbol from $N^{\prime }$.

\medskip

We start with the following theorem where we even obtain a linear tree
structure for the underlying communication graph.

\begin{theorem}
\label{thm:110200} $TCL_{4}(ins_{1}^{1,0},del_{2}^{0,0})=RE.$
\end{theorem}

\begin{proof}
Consider a type-0 grammar $G=\left( N,T,P,S\right) $ in the special Geffert
normal form. We construct a graph-controlled insertion-deletion system%
\begin{equation*}
\Pi =(4,V,T,\{S\},H,1,1,R)
\end{equation*}%
that simulates $G$ as follows. The rules from $P$ are supposed to be labeled
in a one-to-one manner with labels from the set $\left[ 1..|P|\right] $. The
alphabet of $\Pi $ is $V=N\cup T\cup \{p,p^{\prime }\mid p:X\rightarrow
\alpha \in P\}$. The set of rules $R$ of $\Pi $ is defined as follows:

For any rule $p:X\rightarrow bY$ we take the following insertion and
deletion rules into $R$:%
\begin{align*}
p.1.1& :\left( 1,\left( X,p,\lambda \right) _{ins},2\right) & & \\
p.2.1& :\left( 2,\left( p,Y,\lambda \right) _{ins},3\right) & p.2.2& :\left(
2,\left( \lambda ,p^{\prime },\lambda \right) _{del},1\right) \\
p.3.1& :\left( 3,\left( p,p^{\prime },\lambda \right) _{ins},4\right) & 
p.3.2& :\left( 3,\left( p^{\prime },b,\lambda \right) _{ins},2\right) \\
p.4.1& :\left( 4,\left( \lambda ,Xp,\lambda \right) _{del},3\right) \hspace*{%
0.7cm} & &
\end{align*}

For any rule $p:X\rightarrow Yb$ we take the following insertion and
deletion rules into $R$:%
\begin{align*}
p.1.1& :\left( 1,\left( X,p,\lambda \right) _{ins},2\right) & & \\
p.2.1& :\left( 2,\left( p,b,\lambda \right) _{ins},3\right) & p.2.2& :\left(
2,\left( \lambda ,p^{\prime },\lambda \right) _{del},1\right) \\
p.3.1& :\left( 3,\left( p,p^{\prime },\lambda \right) _{ins},4\right) & 
p.3.2& :\left( 3,\left( p^{\prime },Y,\lambda \right) _{ins},2\right) \\
p.4.1& :\left( 4,\left( \lambda ,Xp,\lambda \right) _{del},3\right) \hspace*{%
0.7cm} & &
\end{align*}

For simulating the erasing productions $AB\rightarrow \lambda $ and $%
CD\rightarrow \lambda $ as well as $S^{\prime }\rightarrow \lambda $ we add
the rules $\left( 1,\left( \lambda ,AB,\lambda \right) _{del},1\right) $ and 
$\left( 1,\left( \lambda ,CD,\lambda \right) _{del},1\right) $ as well as $%
\left( 1,\left( \lambda ,S^{\prime },\lambda \right) _{del},1\right) $ to $R$%
..

We claim that $L(\Pi )=L(G)$. We start by proving the inclusion $%
L(G)\subseteq L(\Pi ).$ Let $S\Longrightarrow ^{\ast }uXv\Longrightarrow
ubYv\Longrightarrow ^{\ast }w$ be a derivation of a string $w\in L(G)$. We
show that $\Pi $ correctly simulates the application of the rule $%
p:X\rightarrow bY$. Consider the string $uXv$ in component~$1$. Then, there
is only one possible sequence of applications of rules in $\Pi $: 
\begin{multline*}
\left( 1,uXv\right) \Rrightarrow _{p.1.1}\left( 2,uXpv\right) \Rrightarrow
_{p.2.1}\left( 3,uXpYv\right) \Rrightarrow _{p.3.1}\left( 4,uXpp^{\prime
}Yv\right) \\
\Rrightarrow _{p.4.1}\left( 3,up^{\prime }Yv\right) \Rrightarrow
_{p.3.2}\left( 2,up^{\prime }bYv\right) \Rrightarrow _{p.2.2}\left(
1,ubYv\right) .
\end{multline*}

In a similar way the rules $X\rightarrow Yb$ are simulated: 
\begin{multline*}
\left( 1,uXv\right) \Rrightarrow _{p.1.1}\left( 2,uXpv\right) \Rrightarrow
_{p.2.1}\left( 3,uXpbv\right) \Rrightarrow _{p.3.1}\left( 4,uXpp^{\prime
}bv\right) \\
\Rrightarrow _{p.4.1}\left( 3,up^{\prime }bv\right) \Rrightarrow
_{p.3.2}\left( 2,up^{\prime }Ybv\right) \Rrightarrow _{p.2.2}\left(
1,uYbv\right) .
\end{multline*}

The rules $AB\rightarrow \lambda $ and $CD\rightarrow \lambda $ as well as $%
S^{\prime }\rightarrow \lambda $ are directly simulated by the corresponding
deletion rules $\left( 1,\left( \lambda ,AB,\lambda \right) _{del},1\right) $
and $\left( 1,\left( \lambda ,CD,\lambda \right) _{del},1\right) $ as well
as $\left( 1,\left( \lambda ,S^{\prime },\lambda \right) _{del},1\right) $
from $R$ in component $1$.

Hence, observing that initially $S$ is present in component~$1$, and by
applying the induction argument we obtain that there is a derivation $\left(
1,S\right) \Longrightarrow ^{\ast }\left( 1,w\right) $ in $\Pi $. Thus, we
conclude $L(G)\subseteq L(\Pi )$. For the converse inclusion, it is
sufficient to observe that any computation in $\Pi $ can only be performed
by applying the group of rules corresponding to a production of $G$. Thus,
for any derivation in $\Pi $ a corresponding derivation in $G$ can be
obtained.

Finally we observe that the rules in $R$ induce an even linear structure for
the communication graph, which for short can be represented as follows: 
\begin{equation*}
\circlearrowright 1\rightleftarrows 2\rightleftarrows 3\rightleftarrows 4
\end{equation*}

This observation concludes the proof.
\end{proof}

\bigskip

The next theorem uses one-sided contextual deletion rules.

\begin{theorem}
\label{thm:110110} $GCL_{4}(ins_{1}^{1,0},del_{1}^{1,0})=RE.$
\end{theorem}

\begin{proof}
Consider a type-0 grammar $G=\left( N,T,P,S\right) $ in the special Geffert
normal form. We construct a graph-controlled insertion-deletion system%
\begin{equation*}
\Pi =(4,V,T,\{S\},H,1,1,R)
\end{equation*}%
that simulates $G$ as follows. The rules from $P$ are supposed to be labeled
in a one-to-one manner with labels from the set $\left[ 1..|P|\right] $. The
alphabet of $\Pi $ is $V=N\cup T\cup \{K,K^{\prime }\}\cup \{p \mid p:X\rightarrow \alpha \in P\}$. The set of rules $R$ of $\Pi $ is
defined as follows.

For any rule $p:X\rightarrow bY$ we take the following insertion and
deletion rules into $R$:%
\begin{align*}
p.1.1& :\left( 1,\left( \lambda ,p,\lambda \right) _{ins},2\right) & & \\
p.2.1& :\left( 2,\left( p,X,\lambda \right) _{del},3\right) \hspace*{0.7cm}
& p.2.2& :\left( 2,\left( \lambda ,p,\lambda \right) _{del},1\right) \\
p.3.1& :\left( 3,\left( p,Y,\lambda \right) _{ins},4\right) & & \\
p.4.1& :\left( 4,\left( p,b,\lambda \right) _{ins},2\right) & &
\end{align*}

For any rule $p:X\rightarrow Yb$ we take the following insertion and
deletion rules into $R$:%
\begin{align*}
p.1.1& :\left( 1,\left( \lambda ,p,\lambda \right) _{ins},2\right) & & \\
p.2.1& :\left( 2,\left( p,X,\lambda \right) _{del},3\right) \hspace*{0.7cm}
& p.2.2& :\left( 2,\left( \lambda ,p,\lambda \right) _{del},1\right) \\
p.3.1& :\left( 3,\left( p,b,\lambda \right) _{ins},4\right) & & \\
p.4.1& :\left( 4,\left( p,Y,\lambda \right) _{ins},2\right) & &
\end{align*}

For the erasing production $S^{\prime }\rightarrow \lambda $ we have to add
the rule $\left( 1,\left( \lambda ,S^{\prime },\lambda \right)
_{del},1\right) $, and for the erasing productions $AB\rightarrow \lambda $
and $CD\rightarrow \lambda $ we take the following rules into $R$:%
\begin{align*}
k.1.1& :\left( 1,\left( \lambda ,K,\lambda \right) _{ins},2\right) \hspace*{%
0.7cm} & k.1.2& :\left( 1,\left( \lambda ,K^{\prime },\lambda \right)
_{ins},2\right) \\
k.2.1& :\left( 2,\left( K,A,\lambda \right) _{del},3\right) & k.2.2& :\left(
2,\left( K^{\prime },C,\lambda \right) _{del},3\right) \\
k.2.3& :\left( 2,\left( \lambda ,K,\lambda \right) _{del},1\right) & k.2.4&
:\left( 2,\left( \lambda ,K^{\prime },\lambda \right) _{del},1\right) \\
k.3.1& :\left( 3,\left( K,B,\lambda \right) _{del},2\right) & k.3.2& :\left(
3,\left( K^{\prime },D,\lambda \right) _{del},2\right)
\end{align*}

The simulation of a rule $p:X\rightarrow bY$ of $G$ is performed as follows.
Let the current sentential form be $uXv$. There are several possibilities
here. First, the symbol $p$ is inserted in a context-free manner anywhere in
the string by rule $p.1.1$. After that, either rule $p.2.2$ is applicable,
or, if $p$ was inserted before $X$, rule $p.2.1$ is applicable. In the first
case the string remains unchanged: $uXv$. We remark that this is also the
only evolution if a symbol $q\neq p$ is inserted. In the second case, there
is only one possible further evolution, yielding the desired result $ubYv$
in component $1$: 
\begin{equation*}
\left( 2,upXv\right) \Rrightarrow _{p.2.1}\left( 3,upv\right) \Rrightarrow
_{p.3.1}\left( 4,upYv\right) \Rrightarrow _{p.4.1}\left( 2,upbYv\right)
\Rrightarrow _{p.2.2}\left( 1,ubYv\right)
\end{equation*}

In a similar way the rules $X\rightarrow Yb$ are simulated:%
\begin{equation*}
\left( 2,upXv\right) \Rrightarrow _{p.2.1}\left( 3,upv\right) \Rrightarrow
_{p.3.1}\left( 4,upbv\right) \Rrightarrow _{p.4.1}\left( 2,upYbv\right)
\Rrightarrow _{p.2.2}\left( 1,uYbv\right)
\end{equation*}

Now consider the simulation of the rule $AB\rightarrow \lambda $ (the case
of the rule $CD\rightarrow \lambda $ is treated in an analogous way).
Suppose that $K$ is inserted in a context-free manner in string $u$ by rule $%
k.1.1$ and that we obtain a string $u^{\prime }Ku^{\prime \prime }$ in
component $2$. After that, either rule $k.2.1$ is applicable if $K$ was
inserted before $A$, i.e., $u^{\prime }Ku^{\prime \prime }=u^{\prime
}KAu^{\prime \prime \prime }$, and we obtain the string $u^{\prime
}Ku^{\prime \prime \prime }$ in component $3$, or the string $u$ remains
unchanged and returns to component $1$ by applying rule $k.2.3$. In the
first case, if the first letter of $u^{\prime \prime \prime }$ is not equal
to $B$, the evolution of this string is stopped. Otherwise, if $u^{\prime
\prime \prime }=Bu^{iv}$, rule $k.3.1$ is applied and the string $u^{\prime
}Ku^{iv}$ is obtained in component $2$. Now the computation may be continued
in the same manner and $K$ either eliminates another couple of symbols $AB$
if this is possible, or the string appears in component $1$ without $K$ and
then is ready for new evolutions.

Now in order to complete the proof, we observe that the only sequences of
rules leading to a terminal derivation in $\Pi $ correspond to the groups of
rules as defined above. Hence, a derivation in $G$ can be reconstructed from
a derivation in $\Pi $. Finally we remark that in contrast to the preceding
theorem, the communication graph has no tree structure, yet instead looks
like as follows: 
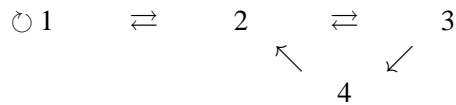
\begin{equation*}
\begin{array}[t]{ccccccccc}
\circlearrowright 1 & \hspace*{0.3cm} & {\LARGE \rightleftarrows } & \hspace*{0.3cm} & 2 & 
\hspace*{0.3cm} & {\LARGE \rightleftarrows } & \hspace*{0.3cm} & 3 \\ 
&  &  &  &  & \nwarrow &  & \swarrow &  \\ 
&  &  &  &  &  & 4 &  & 
\end{array}%
\end{equation*}

These observations conclude the proof.
\end{proof}

\bigskip

The result elaborated above also holds if the contexts for insertion and
deletion rules are on different sides.

\begin{theorem}
\label{thm:110101} $GCL_{4}(ins_{1}^{1,0},del_{1}^{0,1})=RE.$
\end{theorem}

\begin{proof}
We modify the proof of Theorem~\ref{thm:110110} as follows. We replace the
rules $p.2.1$ by the corresponding rules $(2,\left( \lambda ,X,p\right)
_{del},3)$ and the rules $k.2.1$, $k.2.2$, $k.3.1$, and $k.3.2$ by their
symmetric versions. In this case we get a derivation which differs from the
derivation of the previous theorem only by the position of the deleting
symbol, which is inserted in a context-free manner. Hence, the derivations
are equivalent and lead to the same result.
\end{proof}

\bigskip

Finally, we prove that a similar result also holds in the case of
context-free insertions.

\begin{theorem}
\label{thm:200110} $GCL_{4}(ins_{2}^{0,0},del_{1}^{1,0})=RE.$
\end{theorem}

\begin{proof}
Consider a type-0 grammar $G=\left( N,T,P,S\right) $ in the special Geffert
normal form. We construct a graph-controlled insertion-deletion system%
\begin{equation*}
\Pi =(4,V,T,\{S\},H,1,1,R)
\end{equation*}%
that simulates $G$ as follows. The rules from $P$ are supposed to be labeled
in a one-to-one manner with labels from the set $\left[ 1..|P|\right] $. The
alphabet of $\Pi $ is $V=N\cup T\cup \{K,K^{\prime }\}\cup \{p \mid p:X\rightarrow \alpha \in P\}$. The set of rules $R$ of $\Pi $ is
defined as follows.

For any rule $p:X\rightarrow bY$ we take the following insertion and
deletion rules to $R$ (we stress that only one symbol $Y$ is present in the
developing string):%
\begin{align*}
p.1.1& :\left( 1,\left( \lambda ,bp,\lambda \right) _{ins},2\right) & & \\
p.2.1& :\left( 2,\left( p,X,\lambda \right) _{del},3\right) & p.2.2& :\left(
2,\left( Y,p,\lambda \right) _{del},1\right) , \\
p.3.1& :\left( 3,\left( \lambda ,Y,\lambda \right) _{ins},2\right) & &
\end{align*}

For any rule $p:X\rightarrow Yb$ we take the following insertion and
deletion rules into $R$:%
\begin{align*}
p.1.1& :\left( 1,\left( \lambda ,Yp,\lambda \right) _{ins},2\right) \hspace*{%
0.7cm} & & \\
p.2.1& :\left( 2,\left( p,X,\lambda \right) _{del},3\right) & p.2.2& :\left(
2,\left( b,p,\lambda \right) _{del},1\right) \\
p.3.1& :\left( 3,\left( \lambda ,b,\lambda \right) _{ins},2\right) & &
\end{align*}

For the erasing production $S^{\prime }\rightarrow \lambda $ we have to add
the rule $\left( 1,\left( \lambda ,S^{\prime },\lambda \right)
_{del},1\right) $; the erasing rules $AB\rightarrow \lambda $ and $%
CD\rightarrow \lambda $ are simulated by the following rules in $R$:%
\begin{align*}
k.1.1& :\left( 1,\left( \lambda ,K,\lambda \right) _{ins},2\right) \hspace*{%
0.7cm} & k.2.1& :\left( 1,\left( \lambda ,K^{\prime },\lambda \right)
_{ins},2\right) \\
k.2.1& :\left( 2,\left( K,A,\lambda \right) _{del},3\right) & k.2.2& :\left(
2,\left( K^{\prime },C,\lambda \right) _{del},3\right) \\
k.3.1& :\left( 3,\left( K,B,\lambda \right) _{del},4\right) & k.3.2& :\left(
3,\left( K^{\prime },D,\lambda \right) _{del},4\right) \\
k.4.1& :\left( 4,\left( \lambda ,K,\lambda \right) _{del},1\right) & k.4.2&
:\left( 4,\left( \lambda ,K^{\prime },\lambda \right) _{del},1\right)
\end{align*}

The simulation of a rule $X\rightarrow bY$ is performed as follows: 
\begin{equation*}
\left( 1,uXv\right) \Rrightarrow _{p.1.1}\left( 2,ubpXv\right) \Rrightarrow
_{p.2.1}\left( 3,ubpv\right) \Rrightarrow _{p.3.1}\left( 2,ubYpv\right)
\Rrightarrow _{p.2.2}\left( 1,ubYv\right) 
\end{equation*}

Since the rules $p.1.1$ and $p.3.1$ perform a context-free insertion, the
corresponding string can be inserted anywhere. However, if it is not
inserted at the right position, then the computation is immediately blocked,
because the corresponding deletion cannot be performed.

The simulation of a rule $X\rightarrow Yb$ is performed in a similar way: 
\begin{equation*}
\left( 1,uXv\right) \Rrightarrow _{p.1.1}\left( 2,uYpXv\right) \Rrightarrow
_{p.2.1}\left( 3,uYpv\right) \Rrightarrow _{p.3.1}\left( 2,uYbpv\right)
\Rrightarrow _{p.2.2}\left( 1,uYbv\right) 
\end{equation*}%
Observe that $p.2.2$ cannot be used instead of $p.2.1$ because $Y\neq b$.

The erasing rule $AB\rightarrow \lambda \ $is simulated as follows (the
construction for $CD\rightarrow \lambda $ is very similar, using $K^{\prime }
$ instead of $K$, so we omit it here):%
\begin{equation*}
\left( 1,uABv\right) \Rrightarrow _{k.1.1}\left( 2,uKABv\right) \Rrightarrow
_{k.2.1}\left( 3,uKBv\right) \Rrightarrow _{k.3.1}\left( 4,uKv\right)
\Rrightarrow _{k.4.1}\left( 1,uv\right) 
\end{equation*}

The communication graph is identical to the graph in the proof of Theorem 1:
\begin{equation*}
\circlearrowright 1\rightleftarrows 2\rightleftarrows 3\rightleftarrows 4
\end{equation*}

Finally, we remark that only simulations of rules from $G$ as described
above may be part of derivations in $\Pi $ yielding a terminal string;
hence, we conclude $L\left( \Pi \right) =L\left( G\right) $. 
\end{proof}

\section{Conclusions}

In this article we have investigated the application of the mechanism of a
control graph to the operations of insertion and deletion. We gave a clear
definition of the corresponding systems, which is simpler than the one
obtained by using P systems. We investigated the case of systems with
insertion and deletion rules of size $(1,1,0;1,1,0)$, $(1,1,0;1,0,1)$, $%
(1,1,0;2,0,0)$ and $(2,0,0;1,1,0)$ and we have shown that the corresponding
graph-controlled insertion-deletion systems are computationally complete
with only four components, i.e., with the underlying communication graph
containing only four nodes. The case of graph-controlled systems having
rules of size $(2,0,0;2,0,0)$ is investigated in~\cite{KRV09}, where it is
shown that such systems are not computationally complete.

We suggest two directions for the future research. The first one deals with
the number of components needed to achieve computational completeness. The
natural question is if it is possible to obtain similar results with only
three components. The second direction is inspired from the area of P
systems. We propose to further investigate systems where the communication
graph has a tree structure as in Theorem~\ref{thm:110200}. The only known
results so far are to be found in~\cite{KRV09}, but there five nodes were
used. Hence, the challenge remains to decrease these numbers of components.

\bibliographystyle{eptcs}

\end{document}